\documentclass[11pt]{article}

\usepackage[margin=1in]{geometry}
\usepackage{palatino}
\usepackage[small]{caption}
\usepackage{bm}
\usepackage[
  colorlinks,
  linkcolor = blue,
  citecolor = blue,
  urlcolor = blue]{hyperref}
\usepackage{amsmath,amssymb}
\usepackage{amsthm}
\usepackage{mathtools}
\mathtoolsset{centercolon}

\newcommand{\N}{\mathbb{N}}
\newcommand{\C}{\mathbb{C}}

\DeclarePairedDelimiter{\set}{\lbrace}{\rbrace}
\DeclarePairedDelimiter{\abs}{\lvert}{\rvert}

\newcommand{\tp}{^{\textrm{T}}}
\newcommand{\ct}{^{\dagger}}

\newcommand{\ie}{\textit{i.e.}}

\newcommand{\TSAT}{{\texttt{3-SAT}}}
\newcommand{\TCOL}{{\texttt{3-COLORING}}}
\newcommand{\YES}{{\texttt{YES}}}
\newcommand{\NO}{{\texttt{NO}}}
\newcommand{\true}{{\texttt{true}}}
\newcommand{\false}{{\texttt{false}}}
\newcommand{\smx}[1]{\left(\begin{smallmatrix} #1\end{smallmatrix}\right)}

\renewcommand{\vec}[1]{\bm{#1}}

\DeclareMathOperator{\tr}{Tr}

\DeclareMathOperator{\diag}{diag}

\newcommand{\Thm}[1]{\hyperref[thm:#1]{Theorem~\ref*{thm:#1}}}
\newcommand{\Lem}[1]{\hyperref[lem:#1]{Lemma~\ref*{lem:#1}}}
\newcommand{\Cor}[1]{\hyperref[cor:#1]{Corollary~\ref*{cor:#1}}}
\newcommand{\Def}[1]{\hyperref[def:#1]{Definition~\ref*{def:#1}}}
\newcommand{\Obs}[1]{\hyperref[obs:#1]{Observation~\ref*{obs:#1}}}

\newcommand{\Sec}[1]{\hyperref[sec:#1]{Section~\ref*{sec:#1}}}

\newcommand{\Fig}[1]{\hyperref[fig:#1]{Figure~\ref*{fig:#1}}}
\newcommand{\Tab}[1]{\hyperref[tab:#1]{Table~\ref*{tab:#1}}}
\newcommand{\EqRef}[1]{\hyperref[eq:#1]{(\ref*{eq:#1})}}
\newcommand{\Eq}[1]{Equation~\hyperref[eq:#1]{(\ref*{eq:#1})}}

\newtheorem{theorem}{Theorem}
\newtheorem{lemma}{Lemma}
\newtheorem{fact}{Fact}
\newtheorem{cor}{Corollary}

\theoremstyle{definition}
\newtheorem{definition}{Definition}

\def\be{\begin{equation}}
\def\ee{\end{equation}}

\title{Oddities of quantum colorings}
\author{Laura Man\v{c}inska \and David E.~Roberson}

\date{}

\begin{document}

\maketitle

\begin{center}
In memory of Prof.~R\={u}si\c n\v{s} M.~Freivalds
\end{center}\vspace{0.1in}

\begin{abstract}
We study quantum analogs of graph colorings and chromatic number. Initially defined via an interactive protocol, quantum colorings can also be viewed as a natural operator relaxation of graph coloring. Since there is no known algorithm for producing nontrivial quantum colorings, the existing examples rely on ad hoc constructions. Almost all of the known constructions of quantum $d$-colorings start from $d$-dimensional orthogonal representations. We show the limitations of this method by exhibiting, for the first time, a graph with a 3-dimensional orthogonal representation which cannot be quantum 3-colored, and a graph that can be quantum 3-colored but has no 3-dimensional orthogonal representation. Together these examples show that the quantum chromatic number and orthogonal rank are not directly comparable as graph parameters. The former graph also provides an example of several interesting, and previously unknown, properties of quantum colorings. The most striking of these is that adding a new vertex adjacent to all other vertices does not necessarily increase the quantum chromatic number of a graph. This is in stark contrast to the chromatic number and many of its variants. This graph also provides the smallest known example (14 vertices) exhibiting a separation between chromatic number and its quantum analog.
\end{abstract}

%%%%%%%%%%%%%%%%%%%%
\section{Introduction}
\label{sec:Intro}
%%%%%%%%%%%%%%%%%%%%

Graph coloring and chromatic number are standard and well-investigated topics from the early days of graph theory. Given some number of colors, $c$, a $c$-coloring of a graph $G$ is an assignment of colors to vertices of $G$, where adjacent vertices receive different colors. The chromatic number, $\chi(G)$, is the smallest number of colors for which $G$ admits a $c$-coloring. The quantum generalization of graph coloring and chromatic number is comparatively new and was first considered in \cite{Galliard02,Cleve04} and further investigated in \cite{Avis,Cameron07,Fukawa11,SS12,MSS,qhomos}.

In this work we aim to collect some interesting examples of quantum colorings, as well as exhibit some unexpected properties they possess. In many ways, quantum colorings and quantum chromatic number behave very similarly to their classical counterparts. For example, the chromatic and quantum chromatic numbers of complete graphs coincide, a graph is 2-colorable if and only if it is quantum 2-colorable, and both classical and quantum $c$-colorability is preserved under taking subgraphs. However, here we will show that quantum colorings and chromatic number can misbehave in a bizarre manner. Specifically, it is easy to see that adding a new vertex adjacent to all other vertices of a graph causes the chromatic number to increase by one. We will show in \Sec{oddities} that this is not always the case for quantum chromatic number.

There are very few known methods for constructing non-trivial quantum colorings. In fact, except for a few cases, all known constructions of quantum $d$-colorings begin the same way: with an assignment of vectors from $\C^d$ to the vertices of the graph such that adjacent vertices receive orthogonal vectors. Here we show the limits of this construction. In particular, we show that it is possible for a graph to have such an assignment of $d$-dimensional vectors but no quantum $d$-coloring and vice versa. For the former, we will make use of an orthogonality graph, $G_{13}$, constructed from 13 vectors in $\mathbb{R}^3$ with entries from the set $\set{-1,0,1}$. We build on an insight of Burgdorf and Piovesan~\cite{perscom}, who used a computer system for algebraic computations (GAP) to observe that $G_{13}$ is not quantum 3-colorable. Unfortunately, the problem-size is too demanding for the computer to provide a certificate of impossibility. In this work we present an explicit human-readable proof of this fact (see \Thm{notq3col}). For the latter example, we make use of a construction for taking \TSAT{} instances to \TCOL{} instances, which was already used in the context of quantum colorings~\cite{Fukawa11,Ji12}. We show that any graph constructed in this way has a 3-coloring if and only if it has an assignment of vectors in $\C^3$ meeting the above orthogonality condition. Together, these two examples show that quantum chromatic number is incomparable with orthogonal rank, $\xi(G)$, (see \Def{Orank}) which was not known prior to this work.

In general, there is no guaranteed method for finding non-trivial quantum colorings, and the corresponding decision problem is not known to be decidable. The difficulty is that, unlike the classical case in which one can simply search all possibilities, the search space for a quantum $c$-colorings is not only infinite, but also not compact. Because of this there are few nontrivial examples, and relatively little is known. This work provides new examples of quantum colorings which exhibit unexpected properties. The graph $G_{13}$ is particularly interesting as it is the first witness of several such properties.

This paper is organized as follows. In \Sec{Notation} we briefly explain the used notation. Next, in \Sec{Qchrom} we formally introduce quantum colorings and chromatic number, $\chi_q(G)$. We explain the known general constructions of quantum colorings in \Sec{Constructions}. Then we proceed to define graph $G_{13}$ and discuss its classical and quantum parameters in \Sec{G13}. Most notably, we show that $\chi_q(G_{13}) = 4$ (see \Thm{notq3col}) which then gives us a separation $\chi_q(G_{13})>\xi(G_{13})$. Finally, in \Sec{Incomp}, we use \TSAT{} to \TCOL{} reduction to produce a graph $H$ with $\xi(H)>\chi_q(H)$.

%%%%%%%%%
\subsection{Notation}
\label{sec:Notation}
%%%%%%%%%
We use boldface letters, such as $\vec{r}$, to differentiate vectors from other symbols such as vertices of a graph. We use $\vec{r}\ct$, $\vec{r}\tp$, and $\vec{r}^*$ to denote complex conjugate transpose, transpose, and entry-wise complex conjugation of the vector $\vec{r}$. We use $\C^{n\times n}$ to denote the set of all $n\times n$ matrices with complex entries and we write $A \succeq B$ to indicate that the operator $A-B$ is positive semidefinite. Finally we use $[k]$ to denote the set $\set{1,\dotsc, k}$.

\paragraph{Graph terminology.} For us a graph $G = (V,E)$ consists of a finite set, $V$, of vertices and a set, $E$, of unordered pairs of these vertices which we call edges. Thus we only consider finite simple graphs without loops or multiple edges. We will often use $V(G)$ and $E(G)$ to denote the vertex and edge sets of the graph $G$ respectively. We say that two vertices $u,v \in V(G)$ are adjacent if $\{u,v\} \in E(G)$, and we denote this by $u \sim v$. We refer to the vertices adjacent to a given vertex $v$ as the \emph{neighbors} of $v$. A \emph{clique} in a graph $G$ is set of pairwise adjacent vertices in $G$, and an \emph{independent set} is a set of pairwise non-adjacent vertices. The size of the largest clique (independent set) is known as the \emph{clique number} (\emph{independence number}) of $G$ and is denoted by $\omega(G)$ and $\alpha(G)$ respectively. The complement of the graph $G$, denoted $\overline{G}$, is the graph with vertex set $V(G)$ such that two vertices are adjacent in $\overline{G}$ if and only if they are \emph{not} adjacent in $G$. An \emph{apex} vertex of $G$ is a vertex that is adjacent to all other vertices of $G$.

%%%%%%%%%%%%%%%%%%%%
\section{Quantum colorings and chromatic number}
\label{sec:Qchrom} 
%%%%%%%%%%%%%%%%%%%%

%
There are two equivalent perspectives that one can take to define quantum coloring. Initially, quantum colorings were introduced as quantum strategies for a certain nonlocal game~\cite{Galliard02,Cleve04,Cameron07}. We give a brief explanation of this game, but refer the reader to~\cite{Cameron07} for further details. The overall idea of the graph coloring game is that two collaborating provers, Alice and Bob, are attempting to convince a verifier that a given graph admits a $c$-coloring. At the start of the game the verifier selects two vertices $u, v \in V(G)$ that are either equal or adjacent and sends one to Alice and the other to Bob. Without communicating, each of the provers must respond with one of the $c$ colors. Alice and Bob win the game if they have responded with distinct colors in the case $u$ and $v$ were adjacent and the same color in the case $u$ was equal to $v$. The graph $G$ and the number of colors $c$ is known to all the parties, and the provers can use this knowledge to agree on a strategy beforehand. A winning strategy is one which allows the provers to win with certainty no matter which two vertices the verifier selected. Just as in any nonlocal game, neither private nor shared randomness can increase the provers' chances of winning~\cite{Cleve04}. Therefore, classical strategies essentially correspond to a pair of functions $\alpha: V(G) \to [c]$ and $\beta: V(G) \to [c]$ which Alice and Bob use to map the received vertex to the color which they send back to the verifier. It is not hard to see that in order to win with certainty Alice and Bob must choose $\alpha=\beta$ which is also a valid $c$-coloring of $G$. Thus, the existence of a perfect classical strategy for this game is equivalent to the $c$-colorability of $G$. It then follows that the chromatic number $\chi(G)$ is the smallest $c\in\N$ for which there exists a perfect classical strategy for the corresponding coloring game. In analogy, a quantum $c$-coloring of $G$ is an entanglement-assisted strategy which allows the players to win this game with certainty and the quantum chromatic number is defined as
\be 
  \chi_q(G) := \min\set{c\in\N :\text{there exists a perfect quantum strategy for $c$-coloring the graph $G$}}.
\label{eq:Qchrom}
\ee
It has been shown \cite{Cameron07} that perfect quantum strategies for this graph coloring game can always be chosen to take a specific form. Specifically, the following three conditions hold
\begin{itemize}
\item it suffices for the provers to share a specific type of entangled state which is known as a maximally entangled state, 
\item the provers only need to perform projective measurements consisting of same-rank measurement operators, and
\item Alice's measurement operators are complex conjugate to those of Bob. 
\end{itemize}
This simplification allows us to reformulate the existence of a quantum $c$-coloring in a purely combinatorial manner and discuss quantum coloring and chromatic number without defining quantum strategies in their full generality.

\begin{definition}\label{def:qcol}
A \emph{quantum $c$-coloring} of a graph $G$ is a collection of $d$-dimensional orthogonal projectors $\big(v_i : v\in V(G), i\in [c] \big)$ where
\begin{align}
  \sum_{i\in[c]} v_i &= I_d \quad \text{for all vertices $v\in V(G)$ and } 
  && \text{(completeness)}\label{eq:Complete}\\
  v_i w_i &= 0\footnotemark \quad \text{for all $v \sim w$ and all $i\in[c]$.}  
  && \text{(orthogonality)} \label{eq:Adj}
\end{align}\footnotetext{Note that $v_iw_i = 0$ is equivalent to $\tr(v_iw_i) = 0$ since these are positive semidefinite operators.}
The \emph{quantum chromatic number} $\chi_q(G)$ is the smallest $c\in \N$ for which the graph $G$ admits a quantum $c$-coloring in some dimension $d>0$ (this is consistent with \Eq{Qchrom}).
\label{def:Qchrom}
\end{definition}

According to the above definition, any classical $c$-coloring can be viewed as a $1$-dimensional quantum coloring, where we set $v_i = 1$ if vertex $v$ has been assigned color $i$ and we set $v_i = 0$ otherwise. Therefore, quantum coloring is a relaxation of the classical one and for any graph $G$ we have $\chi(G) \ge \chi_q(G)$. The surprising bit is that quantum chromatic number can be strictly and even exponentially smaller than chromatic number for certain families of graphs~\cite{Buhrman98,Brassard99,Avis}. Since we are mostly interested in quantum $c$-colorings for $c<\chi(G)$, we refer to quantum $c$-colorings with $c\ge \chi(G)$ as trivial.

It is worth noting that all graphs which are quantum 2-colorable are also classically 2-colorable. Indeed, if $u$ and $v$ are adjacent vertices of $G$, then any quantum 2-coloring satisfies $u_1v_1 = u_2v_2 = 0$, $u_1+u_2 = I$, and $v_1+v_2 = I$. From this it follows that $u_1 = u_1(v_1+v_2) = u_1v_2 = (u_1+u_2)v_2 = v_2$. Therefore, if $G$ contains any odd cycle and $u$ is a vertex in such a cycle, then $u_1 = u_2$. This contradicts the fact that $u_1u_2 = 0$ and $u_1+u_2 = I$. So we obtain that any graph $G$ containing an odd cycle cannot be quantum 2-colored. Thus quantum $c$-colorings only become interesting when~$c \ge 3$.

%%%%%%%%%%%%%%%%%%%%
\section{Constructions of quantum colorings}
\label{sec:Constructions}
%%%%%%%%%%%%%%%%%%%%

With a couple of exceptions \cite{Fukawa11,Ji12}, all known nontrivial quantum $c$-colorings arise from real or complex orthogonal representations endowed with certain additional properties.

\begin{definition}
A (complex) \emph{orthogonal representation} of a graph $G$ is an assignment of complex unit vectors of some fixed dimension $d$ to vertices of the graph $G$ where adjacent vertices receive orthogonal vectors. The smallest dimension $d$ in which $G$ admits an orthogonal representation is known as the (complex) \emph{orthogonal rank} of $G$ and denoted as $\xi(G)$. We say that a $d$-dimensional orthogonal representation is \emph{flat} if the entries of all the assigned vectors have the same modulus.
\label{def:Orank}
\end{definition}

From now on we will omit the word ``complex'' when referring to complex orthogonal representations and orthogonal rank. It should be noted though that the value of orthogonal rank can change depending on the underlying field of the ambient vector space. For instance, the complex orthogonal rank can differ from the real orthogonal rank for some graphs $G$. 

We can view orthogonal representations as relaxations of colorings. Indeed, by identifying the $i$th color with the $i$th standard basis vector $\vec{e_i}$, we see that any $c$-coloring of $G$ corresponds to an orthogonal representation of $G$ where we have only used the standard basis vectors. Therefore, any $c$-colorable graph also admits a $c$-dimensional orthogonal representation and we arrive at the following fact.
\begin{fact}
For any graph $G$, we have $\xi(G) \le \chi(G)$.
\label{fact:XiChi}
\end{fact}
For some graphs the above inequality can be strict, \ie,  $\xi(G) < \chi(G)$ and in certain cases this implies that also $\chi_q(G) < \chi(G)$. There are two known constructions allowing to translate a $d$-dimensional orthogonal representation $\varphi: V(G) \to \C^d$ into a quantum $d$-coloring of $G$. The first construction applies if $\varphi$ is flat, while the second construction applies if $d =4,8$ and the assigned vectors are real. Whenever the appropriate conditions are met, these constructions can be used to translate separations of the form $\chi(G)>\xi(G)$ into separations of the form $\chi(G) > \chi_q(G)$. In fact, with the exception of the lesser known work \cite{Fukawa11}, all known explicit separations between quantum and classical chromatic numbers are obtained via one of these two constructions. %It appears to be easier to construct nontrivial orthogonal representations than nontrivial quantum colorings.
The reason for this could be that orthogonal representations are easier to work with, as they involve fewer orthogonality relations than quantum colorings. We now proceed to describe the two constructions.

\paragraph{Construction with a flat orthogonal representation.} Let $\varphi : V(G) \to \C^d$ be a flat orthogonal representation and let $F$ be the $d$-dimensional Fourier matrix with entries $F_{kj} =\omega_d^{kj}/\sqrt{d}$, where  $\omega_d:=\exp(2\pi i /d)$.  For any vertex $v$ of $G$ we consider the matrix $U_v := \sqrt{d}  \diag(\varphi(v)) F$, where $\diag(\vec{x})$ is the diagonal matrix with the components of $\vec{x}$ as its diagonal entries. Since both $F$ and $\sqrt{d} \diag(\varphi(v))$ are unitary matrices, $U_v$ must also be a unitary matrix. Note that we needed the orthogonal representation to be flat to ensure that the matrices $\sqrt{d}  \diag(\varphi(v))$ are unitary.

To obtain a quantum $d$-coloring, we set $v_i$ to be the projection onto the vector 
\begin{equation*}
  \vec{r_{vi}}:=\sqrt{d}  \diag(\varphi(v)) F \vec{e_i}
\end{equation*}
which is just the $i$th column of the unitary matrix $U_v$. We will make use of the fact that
\begin{equation*}
  \vec{r_{vi}}= \sqrt{d} \,  \varphi(v) \circ \vec{f_i},
\end{equation*}
where ``$\circ$'' denotes the entry-wise product and $\vec{f_i}:= F \vec{e_i}$. To verify that the projectors $v_i$ give a valid quantum $d$-coloring, we need to check the completeness and orthogonality conditions,~\EqRef{Complete} and \EqRef{Adj}, from \Def{Qchrom}. The completeness condition $\sum_i v_i = I$, follows directly from the fact that the columns $\vec{r_{vi}}$ of the unitary $U_v$ form an orthonormal basis. To check the orthogonality condition, we need to show that $\tr(v_iw_i) =0 $ for all $i\in[d]$ and all adjacent vertices $v$ and $w$. Since $\tr(v_i w_i) = \abs[\big]{\vec{r_{vi}}\ct \vec{r_{wi}}}^2$, it suffices to check that the vectors $\vec{r_{vi}}$ and $\vec{r_{wi}}$ are orthogonal:
\begin{align*}
 \vec{r_{vi}}\ct \vec{r_{wi}} &=  d \,  \big( \varphi(v) \circ \vec{f_i} \big)\ct \big( \varphi(w) \circ \vec{f_i}\big) \\
  &= d\, \vec{1}\ct \Big( \varphi(v)^* \circ \vec{f_i}^* \circ  \varphi(w) \circ \vec{f_i}\Big) \\
   & = d\, \vec{1}\ct \Big( \varphi(v)^* \circ  \varphi(w)\Big)   \\ 
  &= d\, \varphi(v)\ct \varphi(w)   = 0,
\end{align*}
where $\vec{1}$ is the all ones vector, $(\cdot)^*$ denotes entry-wise complex conjugation and we have used the facts that  $\vec{s}\ct \vec{t} = \vec{1}\ct \big( \vec{s}^* \circ \vec{t} \big)$,  $\vec{f_i}^* \circ \vec{f_i} = \vec{1}$, and that $\varphi$ assigns orthogonal vectors to adjacent vertices. 

It is hard to trace the origins of this somewhat folklore construction. It is described in full generality in \cite{Cameron07} but similar a construction had already been used in \cite{Buhrman98,Brassard99}.

\paragraph{Construction with real orthogonal representations in dimension less than eight.}

The authors of~\cite{Cameron07} describe a beautiful construction for converting a \emph{real} orthogonal representation in dimension four into a quantum 4-coloring, and similarly for dimension eight. The constructions use quaternions and octonions respectively. We briefly describe the construction for dimension four, but dimension eight is similar.

Suppose $\vec{r} = (r_0, r_1, r_2, r_3)\tp \in \mathbb{R}^4$ is a unit vector. We aim to use this vector to construct a full orthonormal basis of $\mathbb{R}^4$. To do this, we associate to any vector $\vec{r} \in \mathbb{R}^4$, the \emph{quaternion}
\[q(\vec{r}) = r_0 g_0 + r_1 g_1 + r_2 g_2 +r_3 g_3,\]
where $g_0 = 1, g_1, g_2, g_3$ are the fundamental quaternion units. The usual notation is $1,i,j,k$, but our notation makes it easier to refer to an arbitrary unit. Recall that $g_i^2 = -1$ for $i = 1,2,3$, and that for all \emph{distinct} $i,j \in \{1,2,3\}$ we have $g_ig_j = \pm g_k$ for $k \in \{1,2,3\}\setminus\{i,j\}$. Note that similar rules also hold for the octonions.

Given a vector $\vec{r} \in \mathbb{R}^4$, for any $i \in \{0,1,2,3\}$, let $\vec{r^i} \in \mathbb{R}^4$ be the vector such that $q(\vec{r^i}) = g_iq(\vec{r})$. For instance, $\vec{r^0} = \vec{r}$. We aim to show that $\{\vec{r^0}, \vec{r^1}, \vec{r^2}, \vec{r^3}\}$ is an orthonormal basis. Since multiplying by $g_i$ does not change the magnitude of any coefficient, it is easy to see that $\vec{r^i}$ is a unit vector for all~$i$. Also, left-multiplying $q(\vec{r})$ by $g_i$ swaps the coefficients of $g_0$ and $g_i$ and makes one negative, thus these two coordinates contribute zero to the inner product of $\vec{r}$ and $\vec{r^i}$. The same is true for the remaining two coordinates since $g_ig_j = \pm g_k$ and $g_ig_k = \mp g_j$. Thus $\vec{r}$ and $\vec{r^i}$ are orthogonal, and by similar reasoning all four vectors are pairwise orthogonal and thus form an orthonormal basis. Essentially the same argument works for dimension eight using octonions, except there are four pairs of coordinates that are swapped. In the end, the four vectors obtained from $\vec{r} = (r_0, r_1, r_2, r_3)\tp$ are the columns of the following matrix:
\[
\begin{pmatrix}
r_0 & -r_1 & -r_2 & -r_3 \\
r_1 &  r_0 & r_3 & -r_2 \\
r_2 & -r_3 & r_0 & r_1\\
r_3 & r_2 & -r_1 & r_0
\end{pmatrix}
\]

Now suppose we have a graph $G$ which has a real orthogonal representation $\varphi$ in dimension four. For any vertex $v \in V(G)$, we can construct the orthonormal basis $\{\varphi(v)^0, \varphi(v)^1, \varphi(v)^2, \varphi(v)^3\}$ as above. To obtain a quantum 4-coloring, we let let $v_i$ be the projection onto vector $\varphi(v)^i$. The completeness condition from \Def{Qchrom} is satisfied since the vectors $\varphi(v)^i$ form an orthonormal basis. It remains to show that $u_i$ and $v_i$ are orthogonal whenever $u \sim v$. This corresponds to the vectors $\varphi(u)^i$ and $\varphi(v)^i$ being orthogonal, which holds since these were obtained by permuting and changing some of the signs of the coordinates of $\varphi(u)$ and $\varphi(v)$ respectively.

The same construction works for an eight dimensional real orthogonal representation and so we have the following:

\begin{lemma}{\cite{Cameron07}}\label{lem:realortho2qc}
If a graph $G$ has a real orthogonal representation in dimension four, then $\chi_q(G) \le 4$. Similarly, if $G$ has a real orthogonal representation in dimension eight, then $\chi_q(G) \le 8$.
\end{lemma}

Using the above, one can sometimes use separations between orthogonal rank and chromatic number to construct separations between quantum chromatic number and chromatic number. To see how, suppose $G$ is a graph with a real orthogonal representation in dimension $d$, such that $d < \chi(G)$ and $d \le 8$. By successively adding apex vertices to $G$, we can obtain a graph with a real orthogonal representation in dimension 8 but with chromatic number strictly greater than 8. We can then use the above construction to obtain a quantum 8-coloring of this graph. The same technique works if the orthogonal representation was in dimension at most four. Therefore, we have the following:

\begin{fact}
Suppose $G$ is a graph with a real orthogonal representation in dimension $d \le 4$ and $d < \chi(G)$. Then by successively adding apex vertices to $G$, we can obtain a graph $G'$ such that $\chi_q(G') \le 4 < \chi(G')$. The analogous statement holds for $d \le 8$.
\end{fact}

At this point, it is natural to wonder whether a $d$-dimensional orthogonal representation can always be extended to a quantum $d$-coloring. More generally, one could ask how the quantum chromatic number and orthogonal rank compare. In the next two sections we will show that the two parameters are in fact not comparable.

%%%%%%%%%%%%%%%%%%%%%%%
\section{The curious case of the thirteen vertex graph}
\label{sec:G13}
%%%%%%%%%%%%%%%%%%%%%%%

In this section we define a graph on thirteen vertices, denoted $G_{13}$, which will let us exhibit the promised oddities of quantum colorings. After defining $G_{13}$, we discuss classical graph parameters of $G_{13}$ in \Sec{Classical}. Next we proceed to prove that $\chi_q(G_{13}) = 4$ in \Sec{QChromG13} and end by discussing the unexpected behaviors exhibited by this graph in \Sec{oddities}.

To define $G_{13}$, we consider the nonzero three-dimensional vectors with entries from the set $\set{-1,0,1}$. We identify the vectors $v$ and $-v$, choosing the following set of thirteen representatives:
\begin{align}
 V:= & \ \set[\bigg]{\smx{ 1\\ 0\\ 0},\smx{0\\ 1\\ 0},\smx{0\\ 0\\ 1} } \nonumber\\
 & \ \cup \set[\bigg]{\smx{1\\ 1\\ 0},\smx{1\\ -1\\ 0},\smx{1\\ 0\\ 1},\smx{1\\ 0\\ -1},\smx{0\\ 1\\ 1},\smx{0\\ 1\\ -1}} \label{eq:Vecs}\\
 & \ \cup \set[\bigg]{\smx{1\\ 1\\ 1},\smx{1\\ 1\\ -1},\smx{1\\ -1\\ 1},\smx{-1\\ 1\\ 1}}. \nonumber
\end{align}
Geometrically these vectors can be seen as arising from a three-dimensional cube which is centered at the origin and whose edges have length two. From this viewpoint, the first three vectors correspond to the midpoints of the faces, the next six vectors to midpoints of the edges, and the last four correspond to the vertices of this cube. Using $V$ as the vertex set, we construct $G_{13}$ by making any two orthogonal vertices adjacent. That is, we let $E(G_{13}):=\set{\set{u,v} \subset \binom{V}{2}:  u\tp v=0}$. See \Fig{G13} for a drawing of $G_{13}$, and note that the ten middle vertices form a Petersen graph. 
In \Fig{G13}, we have labeled the vertices by (capital) letters of the alphabet. We will interchangeably refer to the vertices of $G_{13}$ using vectors from the set $V$ defined in \Eq{Vecs} and these letters; the correspondence between these two labellings is explained in the caption of \Fig{G13}.

\begin{figure}
\centering
\includegraphics{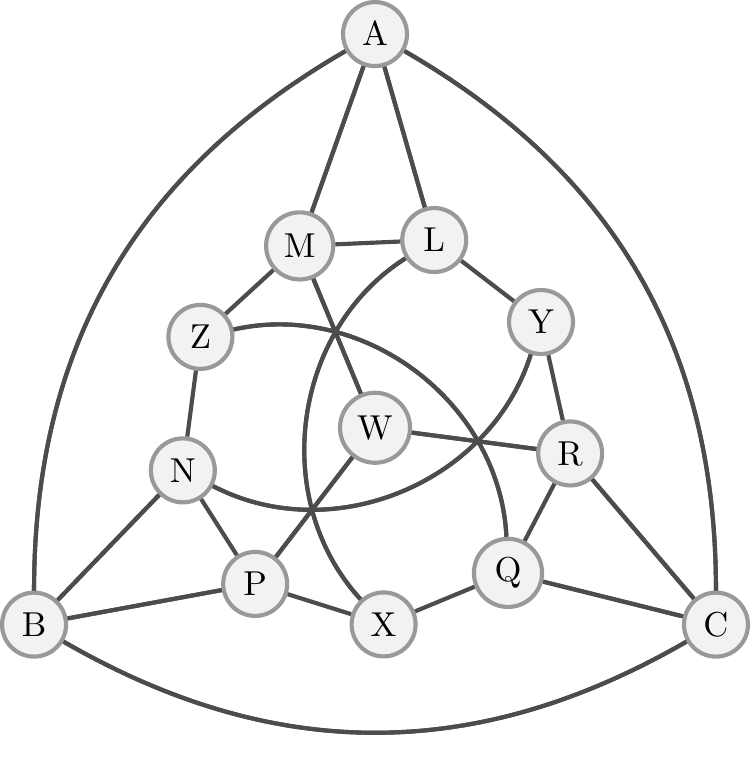}

\caption{The thirteen vertex graph $G_{13}$. The vertices $A,B,C$ correspond to the first three standard basis vectors $\vec{e_1},\vec{e_2},\vec{e_3}$ in that order. The vertices $X,Y,Z,$ and $W$ correspond to vectors from $V$ with no zero entries; the vertex $W$ corresponds to the all ones vector. The remaining six vertices, labelled by letters from the middle of the alphabet, correspond to vectors from $V$ with exactly one nonzero entry.}
\label{fig:G13}
\end{figure}

Another graph that will be useful to us is obtained by adding an apex vertex to $G_{13}$. We refer to this graph as $G_{14}$. We label the vertices of $G_{14}$ that are contained in $G_{13}$ as they are labelled above, and we refer to the apex vertex as $\Omega$. Building on the three-dimensional orthogonal representation of $G_{13}$ given by the vectors in $V$, we can easily construct a four-dimensional orthogonal representation for $G_{14}$. Specifically, to vectors in $V$ we add an additional coordinate which is zero and assign $(0,0,0,1)\tp$ to vertex~$\Omega$.

%%%%%%%%%%%%%%%%%%%%
\subsection{Some classical graph parameters of \texorpdfstring{$G_{13}$}{G13}}
\label{sec:Classical}
%%%%%%%%%%%%%%%%%%%%

We now give a list of some classic graph parameters for the graph $G_{13}$ and its complement $\overline{G_{13}}$. We then proceed with arguments that can be used to obtain the given values. For determining the parameters for $G_{13}$ it will often be useful to consider the corresponding parameter for the Petersen graph, $P$, which can be found in any introductory graph theory textbook.

\begin{center}
\begin{tabular}{l|c|c}
Graph parameter &  \hspace{0.2in} $G_{13}$ \hspace{0.2in} & The Petersen graph $P$\\
\hline
Independence number, $\alpha$ & 5 & 4\\
Chromatic number, $\chi$ & 4 & 3\\
Clique number, $\omega$ & 3 & 2\\
Chromatic number of the complement, $\overline{\chi}$ & 6 & 5\\
Lov\'{a}sz theta number, $\vartheta$ & 5 & 4\\
Lov\'{a}sz theta number of the complement, $\overline{\vartheta}$ & 3 & $5/2$\\
\end{tabular}
\end{center}

\paragraph{Independence number and Lov\'{a}sz theta number.} First, observe that $\alpha(P)=4$. Since the additional vertices in $G_{13}$ form a triangle, we have that $\alpha(G_{13})\le 5$. Finally, since $\set{A,W,X,Y,Z}$ is an independent set we obtain that $\alpha(G_{13})=5$. In fact there are only three different independent sets of size five and they take the form $\set{W,X,Y,Z,v}$, where $v\in\set{A,B,C}$. 

Since $\alpha(G) \le \vartheta(G)$ for any graph, the above implies that $\vartheta(G_{13}) \ge 5$. Next, it is known that $\vartheta(P) =4$ and $\vartheta(K_3) = 1$, so we can upper-bound $\vartheta$ using a similar argument as in the case of $\alpha$. 
Indeed, we get that $\vartheta(G_{13}) \le 5$, by using the following two well-known properties of $\vartheta$. Firstly, $\vartheta$ of the disjoint union of graphs is equal to the sum of $\vartheta$ of the components and secondly, $\vartheta$ cannot increase by adding edges.

\paragraph{Chromatic number.} We start by noting that the four independent sets $\set{A,W,X,Y,Z},\set{B,L,R}$, $\set{C,M,N}$, and $\set{P,Q}$ partition the vertices of $G_{13}$. Therefore, $G_{13}$ is 4-colorable. So to establish that $\chi(G_{13}) = 4$, it remains to argue that it cannot be 3-colored. For contradiction, assume that a valid 3-coloring exists. Without loss of generality assume that we color the vertices $A,B,$ and $C$ with colors 1, 2, and 3 respectively. 
Let $u\in\set{L,M}$ be the vertex colored with color 3, let $v\in\set{N,P}$ be the vertex colored with color 1, and let $w\in\set{Q,R}$ be the vertex colored with color 2. Now note that all three vertices $\overline{u}\in\set{L,M}\setminus\set{u}$, $\overline{v}\in\set{N,P}\setminus\set{v}$, and $\overline{w}\in\set{Q,R}\setminus\set{w}$ also have received different colors. To complete the argument it remains to note that either all three vertices $u,v$ and $w$ or all three vertices $\overline{u},\overline{v}$, and $\overline{w})$ are adjacent to a vertex $z\in\set{W,X,Y,Z}$. Therefore, it is not possible to color $z$ with any of the three colors and we have reached a contradiction.

\paragraph{Clique number and the Lov\'{a}sz theta of the complement.} By inspection, we see that the clique number $\omega(G_{13})=3$. This implies that $\overline{\vartheta}(G_{13}) \ge 3$. On the other hand it is known that $\overline{\vartheta}(G) \le \xi(G)$ for any graph $G$, and so $\overline{\vartheta}(G_{13}) = 3$ as, by construction, $G$ has a three-dimensional orthogonal representation.

\paragraph{Chromatic number of the complement.} The Petersen graph has no triangles and therefore its complement has independence number 2. This implies that the chromatic number of the complement of the Petersen graph is at least 10/2 = 5 and that the color classes in any 5-coloring form a perfect matching in the Petersen graph. Since the Petersen graph has a perfect matching, its complement can indeed by 5-colored. On the other hand, to 5-color the complement of $G_{13}$, we must 5-color the complement of the Petersen graph so that the neighborhood of $A$ in $\overline{G_{13}}$ does not contain vertices of all five different colors, and similarly for $B$ and $C$. However, the only way for this to happen is if each of $\set{L,M}, \set{N,P},$ and $\set{Q,R}$ are a color class of the 5-coloring. In other words, there must be a perfect matching of the Petersen graph containing those three edges. However, this is impossible since the remaining vertices form an independent set. Therefore $\chi(\overline{G_{13}}) \ge 6$. We can 6-color $\overline{G_{13}}$ by first 5-coloring the complement of the Petersen graph and then using an additional color for vertices $A,B,$ and $C$. Therefore $\chi(\overline{G_{13}}) = 6$.

\paragraph{Automorphism group.}
It is easy to see that consistently permuting the coordinates of vectors from $V$ yields automorphisms of $G_{13}$, and there are $6$ possible permutations. When we permute coordinates, we will sometimes need to multiply a vector by $-1$ to recover one of our original thirteen representatives. For instance, if we switch the first and second coordinates, the vector $(1,-1,0)\tp$ becomes $(-1,1,0)\tp$ which is not among the vectors in $V$. We can also multiply any of the coordinates of all the vectors by $\pm1$ to obtain an automorphism. Again, some vectors will have to additionally be multiplied by $-1$ so that we are always using the vectors from the set $V$. Since multiplying all three coordinates by $\pm 1$ gives the identity automorphism, this gives us a total of $6\times (2\times 2 \times 2)/2 = 24$ automorphisms. It turns out that this is the total size of the automorphism group $\mathrm{Aut}(G_{13})$. To see that there are no more automorphisms, first note that vertices $A,B,$ and $C$ are the only ones contained in two triangles, and thus any automorphism fixes these three vertices set-wise. Also, using the automorphisms described above we can permute these three vertices in any way we like (by permuting coordinates).

Now suppose that $\sigma$ is an automorphism of $G_{13}$. We will show that $\sigma$ must be one of the automorphisms we already know. By the above there exists an automorphism $\pi$ which corresponds to permuting the vector coordinates such that $\pi \circ \sigma$ pointwise fixes $A,B,$ and $C$. It follows that $\{N,P\} = \{(1,0,1)\tp,(1,0,-1)\tp \}$ and $\{Q,R\} = \{(1,1,0)\tp,(1,-1,0)\tp\}$ are both set-wise fixed by $\pi \circ \sigma$. By multiplying the second and the third coordinates by $\pm 1$ independently, we can permute $N$ and $P$ arbitrarily as well as $Q$ and $R$. Moreover, these automorphisms fix each of $A,B,$ and $C$. Thus, from our known set of automorphisms, we can select an automorphism $\gamma$ such that $\gamma \circ \pi \circ \sigma$ fixes each of $A,B,C,N,P,Q,$ and $R$. It is easy to see that each of $W,X,Y$ and $Z$ are adjacent to exactly two vertices among $\{N,P,Q,R\}$, and that this pair of vertices is different for each of $W,X,Y$ and $Z$. Therefore, $W,X,Y$ and $Z$ must be fixed by $\gamma \circ \pi \circ \sigma$. From here it easily follows that the remaining two vertices, $L$ and $M$, are fixed as well. Thus $\gamma \circ \pi \circ \sigma$ is the identity automorphism and so we have that $\sigma = (\gamma \circ \pi)^{-1}$ which is one of the automorphisms we already knew about.

There is another fact about the automorphisms of $G_{13}$ that we will need for the proof of \Thm{notq3col}. Note that by permuting the coordinates, we can fix the vertex $W$ and permute the vertices $X,Y,Z$ in any way. Since we can also map any of $X,Y,$ and $Z$ to $W$ by multiplying the appropriate coordinate by $-1$, we have the following fact:

\begin{fact}\label{fact:autos}
Any ordered pair of distinct vertices of $\{W,X,Y,Z\}$ can be mapped to any other such pair by an automorphism of $G_{13}$.
\end{fact}

%%%%%%%%%%%%%%
\subsection{The graph \texorpdfstring{$G_{13}$}{G13} is not quantum 3-colorable}
\label{sec:QChromG13}
%%%%%%%%%%%%%%

The fact that $G_{13}$ is not quantum 3-colorable was first observed by Burgdorf and Piovesan~\cite{perscom} using the GAP system for algebraic computations. In particular, they used the GAP package GBNP made for computing Groebner bases of ideals of non-commutative polynomials. To apply this to quantum coloring, we encode the conditions from \Def{qcol} in a set of polynomials~$\mathcal{P}$. For example, the completeness condition in \Eq{Complete} can be encoded as the polynomial $\sum_i v_i - 1$. In addition, since the operators $v_i$ in  \Def{qcol} need to be projectors, we can also add the polynomials of the form $v_i^2-v_i$ to our set $\mathcal{P}$. 
Next, we consider the ideal $\mathcal{I}$ generated by $\mathcal{P}$, where we treat the $v_i$ as non-commutative variables.  If $\mathcal{I}$ contains the identity element (constant polynomial 1), then we can express the identity as a sum of polynomials $q_k = l_k \cdot p_k \cdot r_k$, where $p_k\in\mathcal{P}$. Since a quantum coloring is an assignment of orthogonal projections to the variables $v_i$ which makes all polynomials $p\in\mathcal{P}$ to evaluate to zero, $\mathcal{I}$ containing identity implies that $I=0$, which is a contradiction. So if the ideal $\mathcal{I}$ contains the identity element, then the graph cannot be quantum colored with the specified number of colors. Given a set of polynomials $\mathcal{P}$, we can use the GBNP package to find a Groebner basis of the ideal $\mathcal{I}$ generated by $\mathcal{P}$. If the identity element belongs to $\mathcal{I}$, then this basis would only contain identity.

In principle the GBNP package can also provide an expression for each of the returned Groebner basis elements in terms of the original polynomials from $\mathcal{P}$. Thus, if the identity element is in the ideal, the package will return an expression for the identity element in terms of the original polynomials, which in principle constitutes an explicit proof. However, in practice this can only be done for small examples. For instance, we were not able to do this even for quantum 3-coloring of a 5-cycle. Even if we used more computing power, it is likely that the expression would be too long for any practical use. For example, even though we can easily establish the impossibility of quantum 3-coloring $K_4$, the GBNP package returns an expression with over 7000 monomials. 

It is also worth noting that the identity element never seems to be in the ideal when considering quantum $c$-coloring for $c \ge 4$, even for large complete graphs. This does not mean that every graph is quantum 4-colorable, since the fact that the identity element is not in the ideal only implies that there exists a solution for the $v_i$ in some primitive algebra, whereas we require a solution in $\C^{d \times d}$ for some $d \in \mathbb{N}$.

Our main contribution in this section is an explicit proof that $G_{13}$ is not quantum 3-colorable. In order to do this we will need some lemmas, the first of which was proven by Ji~\cite{Ji12}, and says that the operators assigned to adjacent vertices in a quantum 3-coloring always commute. We provide a slightly shorter proof of our own.

\begin{lemma}\label{lem:commute}
Suppose $G$ is a graph with adjacent vertices $u$ and $v$. Then in any quantum 3-coloring of $G$, we have that $u_i v_j = v_j u_i$ for all $i,j \in [3]$.
\end{lemma}
\begin{proof}
For $i = j$ the claim follows trivially from the requirement that $u_iv_i = 0 = v_iu_i$ for any quantum 3-coloring of $G$. For the case $i \ne j$, assume without loss of generality that $i = 1$ and $j = 2$. We have that
\[u_1v_2 = u_1v_2(u_1 + u_2 + u_3) = u_1v_2u_1 + u_1v_2u_3,\]
and
\[u_1v_2u_3 = u_1(I - v_1 - v_3)u_3 = u_1u_3 = 0.\]
Therefore, $u_1v_2 = u_1v_2u_1$ and taking conjugate transpose shows that $v_2u_1 = u_1v_2u_1$.
\end{proof}

We also need the following which is proved for $c = 3$ in~\cite{Ji12}.

\begin{lemma}\label{lem:clique}
Let $G$ be a graph and suppose $S$ is a clique of size $c$ in $G$. Then in any quantum $c$-coloring of $G$,
\[\sum_{v \in S} v_i = I \text{ for all } i \in [c].\]
\end{lemma}
\begin{proof}
Note that $u_iv_i = 0$ for all distinct $u,v \in S$ by the conditions of quantum coloring. This implies that $\sum_{v \in S} v_i \preceq I$ for all $i \in [c]$. Next, noting that $\sum_{i \in [c]} v_i = I$ for all $v \in S$, we see that
\[cI = \sum_{v \in S} \sum_{i \in [c]} v_i = \sum_{i \in [c]} \sum_{v \in S} v_i \preceq cI,\]
as $|S| = c$. Since we must have equality throughout, we see that $\sum_{v \in S} v_i = I$ for all $i \in [c]$.
\end{proof}

The next lemma shows that if two vertices share a common neighbor, then certain products of the operators associated to those vertices in a quantum 3-coloring are zero.

\begin{lemma}\label{lem:dist2}
Let $G$ be a graph with vertices $u$ and $v$ that share a neighbor. For any quantum 3-coloring of $G$, if $\{i,j,k\} = \{1,2,3\}$ then 
\[u_iv_ju_k = 0.\]
\end{lemma}
\begin{proof}
Let $w \in V(G)$ be the common neighbor shared by $u$ and $v$. By \Lem{commute}, we have that the operator $w_i$ commutes with both the operators $u_j$ and $v_j$ for all $i,j \in [3]$. Of course we also have that $u_iw_i = 0 = v_iw_i$ for all $i \in [3]$. Using these and the fact that $\sum_{i=1}^3 w_i = I$, it is easy to see that if $\{i,j,k\} = \{1,2,3\}$, then
\begin{align*}
u_iv_ju_k &= u_iv_ju_k(w_i + w_j + w_k) \\
&= u_iv_ju_kw_i + u_iv_ju_kw_j + u_iv_ju_kw_k \\
&= [u_iw_i]v_ju_k + u_i[v_jw_j]u_k + u_iv_j[u_kw_k] \\
&= 0
\end{align*}
\end{proof}

The reader may have noticed that in the second to last line of the equation in the above proof, we put square brackets around the expressions $u_iw_i, v_jw_j,$ and $u_kw_k$. This was to point out that these expressions are easily seen to be equal to zero. For the rest of this section, we will put square brackets around expressions that we know are equal to zero, either by elementary arguments or through the use of the lemmas presented here.

The following lemma requires a little more work, and it specifically concerns the graph $G_{13}$. 

\begin{lemma}\label{lem:XW}
In any quantum 3-coloring of $G_{13}$, we have the following:
\begin{align*}
&X_i A_j W_i = X_i C_j W_i = 0 \ \text{ for } i \ne j, \text{ and } \\
&X_iW_i = X_i A_i W_i = X_i C_i W_i \ \text{ for all } i \in [3].
\end{align*}
\end{lemma}
\begin{proof}
We prove that $X_i A_j W_i = 0$ first. Since the vertices $A, L$, and $M$ form a triangle in $G_{13}$, we have that $A_i + L_i + M_i = I$ by \Lem{clique}. Also, note that all these operators commute with $A_j$ by \Lem{commute} and that $X_iL_i = 0 = M_iW_i$, since $X \sim L$ and $M \sim W$. Therefore, if $i \ne j$ then
\begin{align*}
X_iA_jW_i &=  X_iA_j(A_i + L_i + M_i)W_i \\
&= X_i[A_jA_i]W_i + X_iA_jL_iW_i + X_iA_j[M_iW_i] \\
& = 0 + [X_iL_i]A_jW_i + 0 \\
&= 0.
\end{align*}
The same argument, but using $C_i + Q_i + R_i = I$, shows that $X_iC_jW_i = 0$ if $i \ne j$.

Using the facts that $X_iA_jW_i = 0$ for $i \ne j$ and $A_1 + A_2 + A_3 = I$, we have that
\[X_iW_i = X_i (A_1 + A_2 + A_3)W_i = X_iA_iW_i,\]
and $X_iW_i = X_iC_iW_i$ follows similarly.
\end{proof}

In order to prove that $G_{13}$ is not quantum 3-colorable, we will make frequent use of the above lemmas and so we will not always explicitly point out when we use them. However, after a few applications one easily becomes accustomed to their use.

\begin{theorem}
The graph $G_{13}$ is not quantum 3-colorable.
\label{thm:notq3col}
\end{theorem}
\begin{proof}
Our initial goal is to show that in any quantum 3-coloring of $G_{13}$, the operators $X_1$ and $W_1$ are orthogonal, i.e., that $X_1W_1 = 0$. The first thing to note is that any pair of nonadjacent vertices in $G_{13}$ share a common neighbor. Therefore, we can apply \Lem{dist2} to any such pair of vertices.

We have that $I = B_1 + B_2 + B_3$, and therefore
\[X_1W_1 = B_1X_1W_1 + B_2X_1W_1 + B_3X_1W_1.\]
We will show that each term in the above sum is zero individually.

Since $A,B,$ and $C$ form a triangle in $G_{13}$, by \Lem{clique} we have that $A_3+B_3+C_3 = I$. Therefore, 
\[B_2X_1W_1 = B_2X_1(A_3+B_3+C_3)W_1 = B_2[X_1A_3W_1] + [B_2X_1B_3]W_1 + B_2[X_1C_3W_1].\]
By ignoring the $B_2$ factor on the first and last summands above, applying \Lem{XW} shows that these terms are zero. The second summand is zero since $B_2X_1B_3 = 0$ by \Lem{dist2}, as $B$ and $X$ share $P$ as a neighbor. Thus $B_2X_1W_1 = 0$. Note that the same argument shows that $X_1W_1B_2 = 0$, except that we would use the fact that $B_3W_1B_2 = 0$ in the last step. Similarly, one can show that $B_3X_1W_1 = 0 = X_1W_1B_3$, by placing a factor of $A_2 + B_2 + C_2 = I$ between $X_1$ and $W_1$.

To conclude that $X_1W_1 - 0$, it remains to show that $B_1X_1W_1 = 0$. Since $X_1W_1 = X_1W_1(B_1+B_2+B_3)$ and $X_1W_1B_2 = X_1W_1B_3 = 0$ by the above, it suffices to show that $B_1X_1W_1B_1 = 0$. This we proceed to do. By \Lem{XW}, we have that $X_1W_1 = X_1C_1W_1$. We can therefore rewrite $X_1W_1$ as
\begin{align*}
X_1W_1 &= X_1C_1W_1 \\
&= (I-X_2-X_3)C_1(I-W_2-W_3) \\
&= C_1 - X_2C_1 - X_3C_1 - C_1W_2 - C_1W_3 + X_2C_1W_3 + X_3C_1W_2 + [X_2C_1W_2] + [X_3C_1W_3].
\end{align*}
By \Lem{XW} we have that the last two summands above are equal to zero. Moreover, since $B_1C_1 = C_1B_1 = 0$, by sandwiching the above by $B_1$ we remove all the monomials with a $C_1$ on the far left or the far right. Therefore,
\[B_1X_1W_1B_1 = B_1X_2C_1W_3B_1 + B_1X_3C_1W_2B_1.\]
Using \Lem{clique} and the fact that $A_iB_i = 0$, we can rewrite $C_1$ as
\[C_1 = I - A_1 - B_1 = (I-A_1)(I-B_1) = (A_2+A_3)(B_2+B_3) = A_2B_3+A_3B_2.\]
Therefore, using the fact that $A_i$ and $B_j$ commute for all $i,j \in [3]$ by \Lem{commute}, we have
\begin{align*}
B_1X_1W_1B_1 =& B_1X_2(A_2B_3+A_3B_2)W_3B_1 + B_1X_3(A_2B_3+A_3B_2)W_2B_1 \\
=& [B_1X_2B_3]A_2W_3B_1 + B_1X_2A_3[B_2W_3B_1] \\
&+ B_1X_3A_2[B_3W_2B_1] + [B_1X_3B_2]A_3W_2B_1.
\end{align*}
Since $B,X,$ and $W$ all share $P$ as a neighbor, by \Lem{dist2} we have that all the terms in square brackets above are zero, and therefore $B_1X_1W_1B_1 = 0$. Together with the above, this gives $X_1W_1 = 0$ as desired.

Of course, by symmetry of the colors, we actually have that $X_iW_i = 0$ for all $i \in [3]$. Moreover, by symmetry of the graph $G_{13}$ (see Fact~\ref{fact:autos}), we have that the same holds for any pair of the vertices $W,X,Y,$ and $Z$. Therefore, the operators $W_i,X_i,Y_i,$ and $Z_i$ are mutually orthogonal and thus $W_i + X_i  + Y_i + Z_i \preceq I$ for all $i \in [3]$. Using this we have that
\[4I = \sum_{v \in \{W,X,Y,Z\}} \sum_{i=1}^3 v_i = \sum_{i=1}^3 \sum_{v \in \{W,X,Y,Z\}} v_i \preceq 3I,\]
a clear contradiction. Therefore, no quantum 3-coloring of $G_{13}$ exists.
\end{proof}

We can now determine the value of $\chi_q(G_{13})$.

\begin{cor}
The quantum chromatic number of $G_{13}$ is four.
\end{cor}
\begin{proof}
The above theorem establishes that $\chi_q(G_{13}) \ge 4$, but we already know that $\chi_q(G_{13}) \le \chi(G_{13}) = 4$, and so the corollary follows.
\end{proof}

\subsection{Oddities}
\label{sec:oddities}
%%%%%%%%%%%%%%

Here we point out some interesting properties of the graph $G_{13}$ related to quantum colorings.

\paragraph{Apex vertex.}
The quantum chromatic number $\chi_q$ often behaves similar to the chromatic number $\chi$. For instance, $\chi(K_n) = \chi_q(K_n) = n$ for all complete graphs $K_n$. To our surprise, it turns out that quantum and classical chromatic numbers can behave differently when it comes to apex vertices. It is clear that in any valid coloring of $G$ the color assigned to any apex vertex is different from all the colors assigned to the remaining vertices. Therefore, the addition of an apex vertex increases the chromatic number by one. It turns out that this is not the case for the quantum chromatic number.

\begin{lemma}
There exists a graph for which the addition of an apex vertex leaves the quantum chromatic number unchanged.
\end{lemma}

\begin{proof}
We use $G_{13}$ and $G_{14}$ to establish the lemma. Since we already know that $\chi_q(G_{13}) = 4$ it remains to show that $\chi_q(G_{14})=4$. To this end, note that $G_{14}$ has a four-dimensional orthogonal representation, and so by \Lem{realortho2qc} we have $\chi_q(G_{14}) \le 4$. Since $G_{14}$ contains $G_{13}$ as a subgraph, we have that $\chi_q(G_{14}) \ge \chi_q(G_{13}) = 4$ and hence $\chi_q(G_{14}) = 4$.
\end{proof}

\paragraph{Smallest known example with $\chi>\chi_q$.}

Prior to this work, there were very few \emph{small} graphs with quantum chromatic number known to be strictly less than chromatic number. To the best of our knowledge, the smallest such example was a graph on 18 vertices from~\cite{Cameron07}. With $G_{14}$, we improve this to 14 vertices, and suspect that this is the smallest possible.

\paragraph{Only known example with $\xi(G) < \chi_q(G)$.}

\Thm{notq3col} tells us that the orthogonal rank of $G_{13}$ is strictly less than its quantum chromatic number. To someone who has seen several examples of quantum colorings this may seem odd since almost all known quantum $c$-colorings are constructed from $c$-dimensional orthogonal representations. Thus, one might think that a quantum $c$-coloring can always be constructed from a $c$-dimensional orthogonal representation, but we see here that this is not the case.

Part of the reason that no separation between $\xi$ and $\chi_q$ was previously known is that there are no known graph parameters that lower bound quantum chromatic number but do not also lower bound orthogonal rank, or vice versa. This is related to our next interesting property of $G_{13}$.

%%%%%%%%%%%%%%%%%%%%%%%%%%%
\section{Orthogonal rank and quantum chromatic number are incomparable}
\label{sec:Incomp}
%%%%%%%%%%%%%%%%%%%%%%%%%%%

In the previous section we saw that $\chi_q(G_{13})=4$ and therefore it is possible for orthogonal rank to be strictly smaller than the quantum chromatic number of a graph. In this section our goal is to show that strict inequality can hold in the other direction: orthogonal rank can be strictly smaller than quantum chromatic number. This will mean that the parameters $\xi$ and $\chi_q$ do not satisfy either inequality for all graphs, \ie, they are incomparable.

Both \cite{Fukawa11,Ji12} have investigated quantum 3-colorings of graphs arising from \TSAT{} to \TCOL{} reductions. Here, \TSAT{} is the decision problem whose instances are boolean formulas in conjunctive normal form $\bigwedge_i C_i$, where each of the clauses $C_i$ is a disjunction of exactly three literals (\ie, variables $x_j$ or their negations $\overline{x}_j$). The \YES-instances are the satisfiable formulas, while \NO-instances are the unsatisfiable ones. For example, the formula $(\overline{x}_1 \vee x_2) \wedge (x_1 \vee \overline{x}_2 \vee x_3)$ is a \YES-instance, since it evaluates to \true{} if we set $x_1 = x_2 = x_3=\true$. The instances of \TCOL{} are graphs and the \YES-instances are the graphs which admit a 3-coloring while the \NO-instances do not. For example, the graph $G_{13}$ is a \NO-instance. The basic idea of reductions is to efficiently transform instances of one decision problem to another, so that we can decide the former problem by running an algorithm for solving the latter one. In the case of reducing \TSAT{} to \TCOL{}, we want to translate a \TSAT{} formula $f$ into a graph $G_f$, so that $f$ is satisfiable if and only if $G_f$ is 3-colorable. We now describe one of the standard reductions between the two decision problems which is also used in \cite{Ji12}.  
The graphs arising from this reduction will have a particular structure and not every graph can be obtained in this way.

Given a \TSAT{} formula $f$ with boolean variables $x_1,\dotsc,x_n$ and clauses $C_1,\dotsc, C_m$, we make a graph $G_f$ on $2n+6m+3$ vertices. There are three special vertices labeled, $T$ (for \true), $F$ (for \false), and $B$ (see \Fig{Gadget}).  Each of the variables $x_j$ is represented by two vertices, labeled $x_j$ and $\overline{x}_j$. Intuitively, the former vertex corresponds to setting $x_j = \true$ in the formula $f$, while the latter corresponds to setting $x_j = \false$. Finally, each of the clauses $C_i$ is represented by six vertices three of which are connected to the literals appearing in the clause $C_i$. A gadget encoding clause $x_1 \vee \overline{x}_2 \vee \overline{x}_3$ is shown in \Fig{Gadget} with the six clause-specific vertices enclosed by a dotted box.

\begin{figure}[h]
\centering
\includegraphics{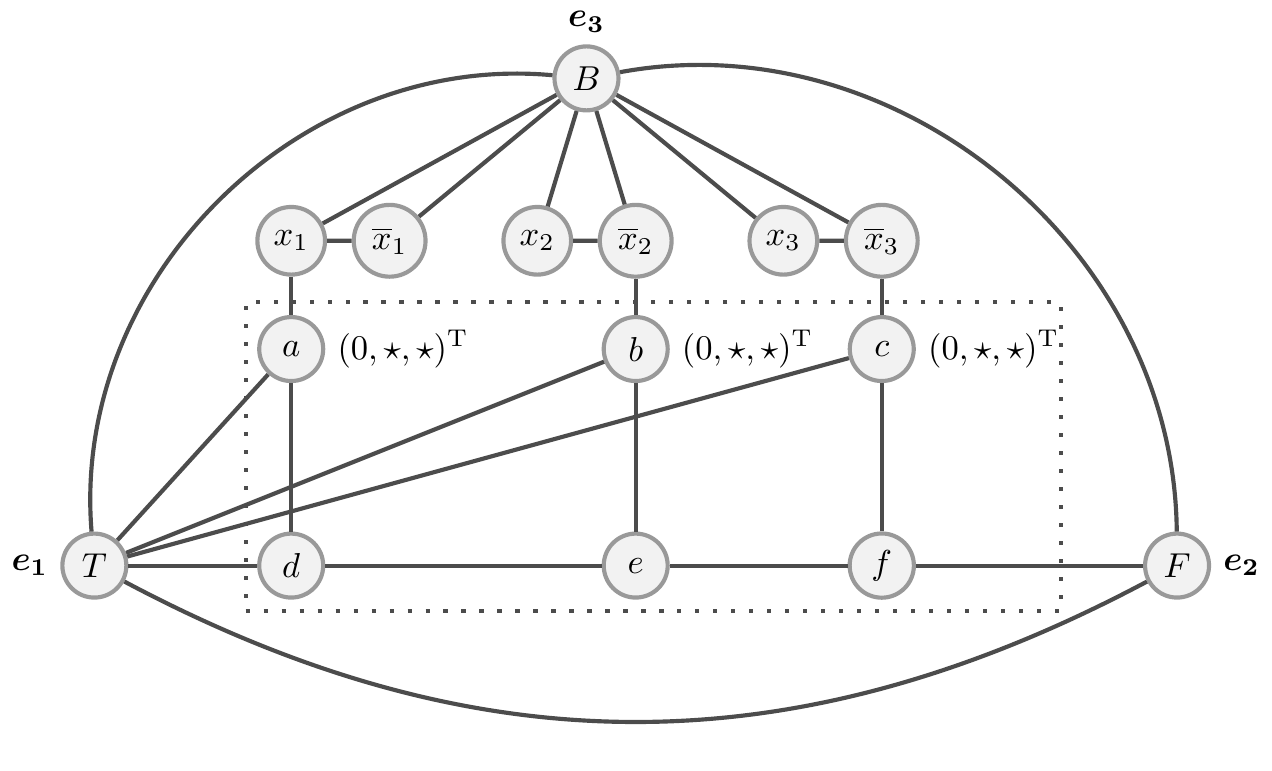}
\caption{Gadget for \TSAT{} to \TCOL{} reduction corresponding to clause $x_1 \vee \overline{x}_2 \vee \overline{x}_3$. Each of the variables $x_i$ appearing in a \TSAT \ formula is represented by two vertices: one corresponding to $x_i$ being true and the other corresponding $\overline{x}_i$ (negation of $x_i$) being true. Each of the clauses is represented by six vertices which are enclosed by a dotted box in the above example. Since variable $x_1$ appears without negation in the clause $x_1 \vee \overline{x}_2 \vee \overline{x}_3$, we connect vertex $a$ to $x_1$; additionally, we connect vertices $b$ and $c$ to $\overline{x}_2$ and $\overline{x}_3$ respectively.}
\label{fig:Gadget}
\end{figure}

It is well known that a \TSAT{} instance $f$ is satisfiable if and only if the graph $G_f$ is 3-colorable. In \cite{Ji12}, the author considers quantum strategies for a nonlocal game $\mathcal{G}_f$ where two provers are aiming to convince the verifier that a given \TSAT{} instance $f$ is satisfiable. They show that perfect quantum strategies for $\mathcal{G}_f$ can be translated into quantum $3$-colorings of the graph $G$ obtained from the \TSAT{} to \TCOL{} reduction we described above. It follows that if $f$ is an unsatisfiable \TSAT{} instance such that there nevertheless exists a perfect quantum strategy for the game $\mathcal{G}_f$, then the graph $G$ will be quantum 3-colorable but not 3-colorable. For a specific example of such a \TSAT{} instance $f$, one can use a \TSAT{} instance $f$ corresponding to the magic square game, as is done in~\cite{Fukawa11}. We therefore have the following:

\begin{fact}
There exists a graph $G$ arising from the \TSAT{} to \TCOL{} reduction for which $\chi_q(G) = 3$ but $\chi(G) > 3$.
\label{fact:Fukawa}
\end{fact}

We remark that a similar observation is made in~\cite{Fukawa11}. However, there the authors employ a slightly different \TSAT{} to \TCOL{} reduction and thus the graph arising from their construction has a slightly different structure.

Using the structure of graphs $G$ arising from the \TSAT{} to \TCOL{} reduction, we argue that any orthogonal representation of such a graph $G$ can be converted into a 3-coloring.

\begin{lemma}
If $G$ is a graph arising from the \TSAT{} to \TCOL{} reduction then $\xi(G) = 3$ implies that $\chi(G) = 3$.
\label{lem:Xi}
\end{lemma}
\begin{proof}
Let $G$ be a graph arising from a \TSAT{} to \TCOL{} reduction with $\xi(G) = 3$. Suppose that vectors $\vec{r_v}\in\C^3$ for $v\in V(G)$ give an orthogonal representation of $G$. Step-by-step we will modify this orthogonal representation so that in the end each of the assigned vectors is one of the three standard basis vectors $\vec{e_1},\vec{e_2},\vec{e_3}\in\C^3$. Obtaining such an orthogonal representation completes the proof as identifying $\vec{e_i}$ with color $i$ gives a 3-coloring of $G$.

Recall that the graph $G$ has three types of vertices: two vertices representing each of the variables $x_k$ (for brevity referred to as $k$ and $\overline{k}$ in this proof), six vertices representing each of the constraints $C_j$ (referred to as $(a,j), (b,j),\dotsc,(f,j)$), and three distinguished vertices $T$, $F$, and $B$. Since the vectors $\vec{r_T},\vec{r_F}$, and $\vec{r_B}$ are mutually orthogonal, there exists a unitary $U$ which maps these three vectors to  $\vec{e_1},\vec{e_2}$, and $\vec{e_3}$, respectively. Since the application of a fixed unitary preserves all the orthogonalities, as our first modification, we apply $U$ to all the assigned vectors $\vec{r_v}$. Now, for any constraint $C_j$, all of the vertices $(a,j), (b,j), (c,j)$ are adjacent to $T$ and hence are assigned vectors of the form $(0,\star, \star)$, where $\star$ denotes an unknown complex number which can be different at each occurrence (see \Fig{Gadget}). Similarly, the vectors $\vec{r_{k}}$ and $\vec{r_{\overline{k}}}$ assigned to vertices representing any of the variables $x_k$ must take the form $(\star,\star,0)$.

Suppose that some vertex $k$ has been assigned a vector $\vec{r_k} = (a,b,0)\tp$, where both of $a,b\in\C$ are nonzero (if only one of $a$ and $b$ are nonzero, then $\vec{r_k}$ is already proportional to either $\vec{e_1}$ or $\vec{e_2}$). To ensure orthogonality,  the vector $\vec{r_{\overline{k}}}$ must equal $(-b^*,a^*,0)\tp$ up to an (irrelevant) overall phase factor. Recall that  the vertices $(a,j),(b,j),(c,j)$ are assigned vectors of the form $(0,\star,\star)$. So if any of these vertices are adjacent to $k$ or $\overline{k}$, then the vectors assigned to them must be further restricted to take the form $(0,0,\star)$. Therefore, setting $\vec{r_{k}}=\vec{e_1}$ and $\vec{r_{\overline{k}}}=\vec{e_2}$ yields a valid orthogonal representation. At this point all but the constraint-specific vertices of $G$ are assigned one of the standard basis vectors. 

Now consider the vertices $(a,j),(b,j),\dotsc, (f,j)$ representing some constraint $C_j$ (see \Fig{Gadget}). Suppose for contradiction that all three of the neighbors of $(a,j), (b,j), (c,j)$ representing variables are assigned $\vec{e_2}$. Then all of the vectors $\vec{r_{(a,j)}}, \vec{r_{(b,j)}}, \vec{r_{(c,j)}}$ are proportional to $\vec{e_3}$. Furthermore, all of the vectors $\vec{r_{(d,j)}}, \vec{r_{(e,j)}}, \vec{r_{(f,j)}}$ must be proportional to either $\vec{e_1}$ or $\vec{e_2}$. However, this is a contradiction since this would provide a 2-coloring of the 5-cycle consisting of vertices $T, (d,j), (e,j), (f,j), F$. Therefore, at least one of the three neighbors of $(a,j), (b,j), (c,j)$ representing variables must be assigned $\vec{e_1}$. At this point the orthogonal representation can be completed using standard basis vectors in the same way as a 3-coloring is constructed in the proof of the classical reduction of \TSAT{} to \TCOL. Briefly, the vertex among $(a,j), (b,j), (c,j)$ whose neighbor was assigned $\vec{e_1}$ should be assigned $\vec{e_2}$, and its neighbor among $(d,j), (e,j), (f,j)$ should be assigned $\vec{e_3}$. The remaining two vertices among $(a,j), (b,j), (c,j)$ should be assigned $\vec{e_3}$ as well. Only two vertices remain and it is easy to see that they can each be assigned one of $\vec{e_1}$ and $\vec{e_2}$. Doing this for every constraint yields an orthogonal representation of $G$ which uses only standard basis vectors thus completing the proof.
\end{proof}

Combining Fact~\ref{fact:Fukawa} with \Lem{Xi} yields the first example of graph whose orthogonal rank is known to exceed its quantum chromatic number.

\begin{cor}
If $G$ is the graph from Fact~\ref{fact:Fukawa}, then $\xi(G)> 3$ but $\chi_q(G)=3$. Hence, $\chi_q(G)<\xi(G)$.
\end{cor}

\section*{Acknowledgments} 
Both authors are thankful to Sabine Burgdorf and Teresa Piovesan for insights and helpful discussion.

\bibliographystyle{alphaurl}
\bibliography{G13.bbl}

\end{document}